\def\year{2020}\relax
\documentclass[letterpaper]{article} 
\usepackage{aaai20}  
\usepackage{times}  
\usepackage{helvet} 
\usepackage{courier}  
\usepackage[hyphens]{url}  
\usepackage{graphicx} 
\usepackage{graphicx}  
\usepackage[utf8]{inputenc}
\usepackage[utf8]{inputenc}
\usepackage{amsfonts}
\usepackage{amsmath}
\usepackage{amsthm}
\usepackage{booktabs}
\usepackage{mathrsfs}
\usepackage{mathtools}
\usepackage{blkarray}
\usepackage{cleveref}
\usepackage{xcolor}
\usepackage{siunitx}
\newtheorem{definition}{Definition}
\newtheorem{theorem}{Theorem}
\newtheorem{lemma}{Lemma}

\newtheorem{proposition}{Proposition}
\newcommand{\AutoAdjust}[3]{\mathchoice{ \left #1 #2  \right #3}{#1 #2 #3}{#1 #2 #3}{#1 #2 #3} }
\newcommand{\Xcomment}[1]{{}}

\newcommand{\InBrackets}[1]{\AutoAdjust{[}{#1}{]}}
\newcommand{\Ex}[2][]{\operatorname{\mathbf E}_{#1}\InBrackets{#2}}
\newcommand{\Prx}[2][]{\operatorname{\mathbf{Pr}}_{#1}\InBrackets{#2}}

\DeclareMathOperator{\argmax}{argmax}
\newcommand{\noaccents}[1]{#1}
\newcommand{\newagentvar}[3][\noaccents]{%
\expandafter\newcommand\expandafter{\csname #2\endcsname}{#1{#3}}%
\expandafter\newcommand\expandafter{\csname #2s\endcsname}{#1{\boldsymbol{#3}}}%
\expandafter\newcommand\expandafter{\csname #2smi\endcsname}[1][i]{#1{\boldsymbol{#3}}_{-##1}}%
\expandafter\newcommand\expandafter{\csname #2sno\endcsname}[1][i]{#1{\boldsymbol{#3}}^{no}}%
\expandafter\newcommand\expandafter{\csname #2no\endcsname}[1][i]{#1{#3}^{no}_{##1}}%
\expandafter\newcommand\expandafter{\csname #2sbno\endcsname}[1][i]{#1{\boldsymbol{#3}}^{bno}}%
\expandafter\newcommand\expandafter{\csname #2bno\endcsname}[1][i]{#1{#3}^{bno}_{##1}}%
\expandafter\newcommand\expandafter{\csname #2i\endcsname}[1][i]{#1{#3}_{##1}}%
\expandafter\newcommand\expandafter{\csname #2j\endcsname}[1][i]{#1{#3}^{##1}}%
\expandafter\newcommand\expandafter{\csname #2ith\endcsname}[1][i]{#1{#3}_{(##1)}}%
}

\newagentvar{report}{b}
\newagentvar{type}{t}
\newagentvar[\bar]{dtype}{t}
\newagentvar{alloc}{x}
\newagentvar[\tilde]{extalloc}{x}
\newagentvar{detalloc}{\alloc^{\textup{det}}}
\newagentvar{randalloc}{\alloc^{\textup{rand}}}
\newagentvar{pay}{p}
\newagentvar[\tilde]{extpay}{p}
\newagentvar{typespace}{T}
\newagentvar{dettypespace}{\typespace^{\textup{det}}}
\newagentvar{randtypespace}{\typespace^{\textup{rand}}}

\newagentvar{map}{\phi}
\newagentvar{intpay}{\rho}
\newagentvar{vertype}{V}
\newagentvar{component}{\omega}

\newcommand{\fixedmenuset}{X_{V_2}}
\newcommand{\feasible}{\mathscr{F}}
\DeclareMathOperator{\supp}{supp}

\DeclareMathOperator{\convhull}{Conv}
\newagentvar{forced}{Force}

\newcommand{\zeroalloc}{\vec{\mathbf{0}}}
\newcommand{\allones}{\vec{\mathbf{1}}}

\newcommand{\singleonevec}{\mathcal{S}_{1}}
\newagentvar{val}{v}
\newagentvar{util}{u}
\newagentvar{payment}{p}
\DeclareMathOperator{\OPT}{OPT}
\DeclareMathOperator{\Conv}{Conv}

\newcommand{\swapfeasible}{\boldsymbol{x}^{\mathrm{ssf}}}

\newcommand{\old}[1]{}
\title{Limitations of Incentive Compatibility on Discrete Type Spaces}
\author{Submission Number: 9394}
\date{May 2019}
\author{Taylor Lundy and Hu Fu\\University of British Columbia\\tlundy@cs.ubc.ca, hufu@cs.ubc.ca\\}

\newcommand{\citet}[1]{\citeauthor{#1} \shortcite{#1}}
\newcommand{\citealp}[1]{\citeauthor{#1} \citeyear{#1}}
\begin{document}

\maketitle
\begin{abstract}  
	In the design of incentive compatible mechanisms, a common approach is to enforce incentive compatibility as constraints in programs that optimize over feasible mechanisms. Such constraints are often imposed on sparsified representations of the type spaces, such as their discretizations or samples, in order for the program to be manageable. In this work, we explore limitations of this approach, by studying whether all dominant strategy incentive compatible mechanisms on a set~$T$ of discrete types can be extended to the convex hull of~$T$.

Dobzinski, Fu and Kleinberg (2015) answered the question affirmatively for all settings where types are single dimensional. It is not difficult to show that the same holds when the set of feasible outcomes is downward closed. In this work we show that the question has a negative answer for certain non-downward-closed settings with multi-dimensional types. This result should call for caution in the use of the said approach to enforcing incentive compatibility beyond single-dimensional preferences and downward closed feasible outcomes.
\end{abstract}

\section{Introduction}
\label{sec:intro}
Mechanism design studies optimization problems with private inputs from strategic agents.  
An agent's input, known as her \emph{type}, is her private information including her valuation for the social outcomes, which are to be decided upon by the mechanism.  
A mechanism needs to solicit such information to achieve certain goals, e.g.\@ maximizing welfare, revenue, surplus or fairness measures.
It needs to provide its participants with correct incentives, via both social outcomes and payments, so that the agents find it in their best interests to reveal their types.  

Dominant strategy incentive compatibility (DSIC) is one of the strongest and most widely used solution concepts that guarantee such incentives.  Under DSIC, every participant, no matter what type she possesses and no matter what types the other participants report to the mechanism, will maximize her utility by revealing her true type.  
Not only is this a strong guarantee for the mechanism designer that true information should be reported and optimized over, it also alleviates the burden of strategizing from the participating agents --- telling truth is a dominant strategy regardless of the other agents' types or strategies.  
Partly thanks to this strong incentive guarantee, the two fundamental auctions, namely, the VCG auction that maximizes social welfare \cite{Vic61,Clarke71,Groves73}, and Myerson's auction that maximizes expected revenue for selling a single item \cite{Myerson81}, have been foundational in both the theory and practice of mechanism design.


As the scope of mechanism design expands beyond the classical settings, incentive compatible mechanisms that are optimal for various objective often lack the simple structures that characterize the VCG and Myerson's mechanisms.  
By and large, there have been two approaches to the design of incentive compatible mechanisms.
The first approach focuses on classes of mechanisms that, by their simple structures, have obvious incentive guarantees.  
For example, in a multi-item auction, a sequential pricing mechanism puts prices on items and asks each agent in turn to choose her favorite items that remain; the bidders are not asked about their values, and choosing utility-maximizing items (according to their true values) is the obvious strategy to adopt (see, e.g., \citealp{CHMS10}; \citealp{CMS10}; \citealp{FGL15}).  
Another example is to optimize over parameterized ``VCG-like'' mechanisms which inherit incentive properties from the VCG mechanism (e.g. \citealp{sandholm2015automated}).
This approach is often used to search for mechanisms whose performance is a factor away from being optimal, since the optimal mechanism or its very close approximations are often not within the class of mechanisms being searched over.

The second approach forgoes structures that are easily interpretable, and exhaustively searches for the optimal mechanism.  
This is exemplified by solving mathematical programs (typically linear or convex programs) whose feasible regions contain the set of all incentive compatible mechanisms (see, e.g. \citealp{conitzer2002complexity}; \citealp{DFK15}; \citealp{DDT17};\citealp{FH18}).  
Typically, incentive requirements are hardwired as constraints in such programs.  

Difficulty arises in the second approach when one would like to adopt strong incentive guarantees such as DSIC, which need at least one constraint per profile of types to specify.  When the space of possible types is a continuum, this gives rise to uncountably many constraints. While this does not always make the program impossible to solve it considerably complicates the task.  
One way to work around this is to discretize the type space and only impose incentive compatible (IC) constraints on the set of discrete types used to represent the type space.  
Discretization is also embodied in the idea of a given prior distribution over a set of discrete types, on which the optimization can then be based (e.g. \citealp{conitzer2002complexity}; \citealp{DFK15}).
The most common motivation for such prior distributions is that they naturally result from samples (e.g.\@ from past observations or market research) from an underlying distribution, whereas the true distribution may be supported on a continuum of types.
This approach motivates the question we study in this work.


\paragraph{Questions we study.}  In this work we aim to answer the question: when one has a mechanism that is DSIC on a discretized subset of a type space, can one always find a mechanism that has the same behavior on the subset and yet is DSIC on the whole type space?
To make the question more concrete, we study the natural case where the whole type space is the convex hull of the discrete subset.
To make the presentation easier, in the following we denote by~$\typespaces$ the discrete subset of types, and $\convhull(\typespaces)$ its convex hull.

We consider the question a fundamental one for the second approach to mechanism design that we described above.
When one optimizes for a mechanism with IC constraints imposed only on~$\typespaces$, if the resulting mechanism cannot be \emph{extended} to the original type space, it loses incentive guarantees when put to use.

One objection may be that, given a mechanism that is DSIC on~$\typespaces$, one may always run it on $\convhull(\typespaces)$, by restricting the ``bidding language'',  so that a type in $\convhull(\typespaces)$ but not in~$\typespaces$ has to report a type in~$\typespaces$.  
Such mechanisms, however, may lose the incentive guarantee which makes DSIC mechanisms attractive in the first place.
Unless one can show that agents with types not in~$\typespaces$ have a dominant strategy in such mechanisms, such agents need to strategize over which types in~$\typespaces$ to report, depending on the types and strategies of their opponents.
In the scenario where $\typespaces$ is a set of samples from a continuous distribution, the vast majority of types may not be in~$\typespaces$ and have no incentive guarantee, which is clearly undesirable.
In settings where the agents' types are single dimensional, \citet{DFK15} showed that ``restricting the bidding language'' does turn any mechanism DSIC on~$\typespaces$ into a mechanism DSIC on any superset of~$\typespaces$: each type in the superset has a dominant strategy, and by the revelation principle this gives rise to a DSIC mechanism that extends the given mechanism's behavior on~$\typespaces$.
To the best of our knowledge, no such guarantees are known beyond single dimensional settings.

\paragraph{Our Results.} 
For agents with multi-dimensional types, we first give a condition under which any DSIC mechanism on~$\typespaces$ can be extended to a DSIC mechanism on $\convhull(\typespaces)$, via an argument that is different from \citeauthor{DFK15}'s yet still straightforward (Theorem~\ref{thm:swap}).  In particular, the condition is satisfied whenever the set of feasible outcomes is \emph{downward closed} (Theorem~\ref{thm:downward}).  

Our main result, however, is a construction of a set $\typespaces$ of multi-dimensional types and a DSIC mechanism on it, for which we show that no DSIC mechanism on~$\convhull(\typespaces)$ can output the same social outcomes on types in~$\typespaces$.  
The impossibility result stands even if the extension mechanism is allowed to be randomized.  
This shows that, without conditions such as single-dimensional types or downward closed set of feasible outcomes, designing incentive compatible mechanisms by focusing on a discrete subset of types can be a questionable approach to designing mechanisms for the whole type space --- there may not be any mechanism DSIC on the whole type space which behave the same way on the subset.

Near the end, we give a multi-dimensional setting where the expected \emph{revenue} of a mechanism with only correct incentives for a set~$\typespace$ of types can be unboundedly more than the revenue of a mechanism for $\convhull(\typespace)$.
This example is much less involved than our main result, because revenue optimal mechanisms are meaningful only when they do not overcharge any reported type and guarantee non-negative utility.  This constraint can be much more stringent when imposed for all types in~$\convhull(\typespace)$ than for $\typespace$ only.
 
\subsection{Related Works}
\label{sec:related}

For multi-dimensional preferences, an allocation rule for any fixed agent is implementable if and only if it satisfies the so-called \emph{cyclic monotonicity} property \cite{Rochet87}.  When the type space is convex, it turns out the weaker condition of \emph{weak monotonicity} suffices for implementability \cite{SY05,AK14}.  It is notable that the two solution concepts we compare in Section~\ref{sec:mapping} precisely correspond to the case where the type space is convex and that where it is not.  However, nowhere in our arguments do we make use of this beautiful fact.

Another closely related body of work, is the literature on automated mechanism design. In automated mechanism design, mechanisms are optimized for the setting and objective automatically using information about agents' type distributions. When this work was introduced by \citet{conitzer2002complexity} the input for this problem was an explicit description of the agents distributions, however recent work has moved towards replacing this explicit description with samples from the agents type distribution \cite{likhodedov2004methods,likhodedov2005approximating,sandholm2015automated}. Our work highlights how interpolating between discrete samples can effect not only the objective but the implementablility of the mechanism itself. Luckily, the research on sample based automated mechanism design is able to avoid the pitfalls of only having discrete samples. They do this either by optimizing over parameterized families of mechanisms which are guaranteed to be implementable on the entire typespace or by working in settings where the addition of new types has no effect on the objective (i.e. downward closed settings, see \Cref{thm:downward})\cite{sandholm2015automated,guo2010computationally,balcan2018,morgenstern2016learning}. 
However, our work points out some difficulties that might arise if one wishes to take a more general, non-parameterized approach to automated mechanism design in settings which are not downward closed.

 There are a variety of well-studied settings that are not downward closed in which extending a type space to its the convex hull could cause problems.  One commonly studied not downward closed setting arises from the job scheduling problem introduced in \citet{nisan2001algorithmic} and later built upon by \citet{schedule2} and \citet{ashlagi2012optimal}. since in this problem every job must eventually be scheduled and the set of feasible solutions is not downward closed.
 Another example of a non-downward closed setting  is one-sided matching markets in which every agent must be matched with exactly one good. An example of a one-sided matching market is the fair housing allocation studied by \citet{matchmarket}. Finally, the facility location problem from \citet{devanur2005strategyproof} also not downward closed.

\section{Preliminaries}
\label{sec:prelim}
We consider a setting with $N$~agents where each agent~$i$ has a private type $\typei$ from her type space $\typespacei \subseteq \mathbb R_+^m$.
The type profile $\types = (\typei[1], \ldots, \typei[N])$ denotes the vector of all agents' types, from the joint type space  $\typespaces \coloneqq \prod_i \typespacei$.

We adopt the standard shorthand notation to write $\typesmi \coloneqq (\type_{1}, \ldots, \type_{i-1}, \type_{i+1}, \ldots, \type_{n})$ from $\typespacesmi \coloneqq \prod_{j \neq i} \typespace_{j}$.

An outcome (or, interchangeably, allocation) for agent~$i$ lies in $\mathbb R_+^m$; for an outcome $\alloci$, the agent with type~$\typei$ has value $\langle \typei, \alloci \rangle = \sum_{j = 1}^m \type_{ij} \alloc_{ij}$.  
A social outcome is denoted by a vector $(\alloci[1], \ldots, \alloci[N]) \in \mathbb R_+^{mN}$.  
The set of all feasible social outcomes (or allocations) is denoted $\feasible \subseteq \mathbb R_+^{mN}$.

For example, in a single-item auction, $m = 1$, each $\typei \in \mathbb R_+$ represents agent~$i$'s value for the item, and $\feasible \subseteq \mathbb R_+^N$ is the all zero vector (representing not selling) and the $N$ standard bases (each representing selling to a corresponding agent).   
As another example, in a $m$-unit auction with unit-demand buyers, $\feasible \subseteq \mathbb R_+^{N}$ is the set of all integral points in $\{(\alloci[1], \ldots, \alloci[N]) \in \mathbb R_+^{mN} \mid \sum_{i = 1}^N \alloci[ij] \leq 1, j = 1, 2, \ldots, m\}$.

\paragraph{Mechanisms.}
A (direct revelation) mechanism consists of an \emph{allocation rule} $\allocs: \typespaces \to \feasible$ and a \emph{payment rule} $\pays: \typespaces \to \mathbb R_+^N$.  
The mechanism elicits type reports from the agents, and on 
reported type profile $\types$, decides on an allocation $\allocs(\types) \in \feasible$, with each agent~$i$ making a payment of $\payi(\types)$.  
In general, allocation rules can be randomized, in which case $\allocs(\types)$ is a randomized variable supported on $\feasible$.
$\allocs(\cdot)$ induces allocation rule $\alloci(\cdot)$ for each agent~$i$: for all $\types$, $\alloci(\types) \in \mathbb R_+^m$ is the vector consisting of the $[(i-1)m + 1]$-st to the $im$-th coordinates in $\allocs(\types)$.
When $\allocs(\types)$ is a random variable, so are $\alloci(\types)$'s. 

When $\allocs(\types)$ is deterministic, we write $\allocs(\types) = \mathbf{y} \in \feasible$ as a shorthand for $\Prx{\allocs(\types) = \mathbf y} = 1$.

Agents have quasi-linear utilities, that is, when reporting type~$\typei'$, agent~$i$'s utility is $\Ex{\langle \typei, \alloci(\typei', \typesmi) \rangle } - \payi(\typei', \typesmi)$ (where the expectation is taken over the randomness in $\allocs(\types)$.

A mechanism is \emph{dominant strategy incentive compatible} (DSIC) if, for all $\types \in \typespaces$, and for all $\typei' \in \typespacei$, $\Ex{\langle \typei, \alloc_{i}(\typei, \typesmi) \rangle} - \payi(\typei, \typesmi) \geq \Ex{\langle \typei,  \alloc_{i}(\typei', \typesmi) \rangle} - \payi(\typei', \typesmi)$.  
An allocation rule~$\allocs$ is said to be DSIC implementable or simply DSIC if there is a payment rule~$\pays$ such that $(\allocs, \pays)$ is a DSIC mechanism.  In this case, we say $\allocs$ is implemented by payment rule~$\pays$.

\paragraph{Extensions.}

Given a subset $S \subseteq \mathbb R^n$, we denote by $\Conv(S)$ the convex hull of~$S$. 
\begin{definition}
    \label{def:extension}
    An allocation rule $\extallocs: \convhull(\typespaces) \to \feasible$ is an \emph{extension} of an allocation rule $\allocs: \typespaces \to \feasible$ if for all $\types \in \typespaces$, $\extallocs(\types)$ has the same distribution as $\allocs(\types)$.  
    Similarly, a payment rule~$\extpays: \convhull(\typespaces) \to \mathbb R_+^N$ is an extension of payment rule~$\pays: \typespaces \to \mathbb R_+^N$ if for all $\types \in \typespaces$, $\extpays(\types) = \pays(\types)$.
\end{definition}

    In Definition~\ref{def:extension}, if $\allocs(\cdot)$ is deterministic, then $\extallocs(\cdot)$ being an extension simply means $\extallocs(\types) = \allocs(\types)$ for all $\types \in \typespaces$.
    
\paragraph{Downward closed settings.}  The feasible allocation set $\feasible$ is \emph{downward closed} if $\mathbf y \in \feasible$ entails $\allocs \in \feasible$ for all $\allocs \preceq \mathbf y$, where $\allocs \preceq \mathbf y$ denotes $\alloci[j] \leq y_j$ for $j = 1, \cdots, mN$.

\paragraph{Weak monotonicity.} 
A well-known necessary condition for an allocation rule to be DSIC implementable is weak monotonicity:

\begin{definition}
	An allocation rule $\allocs: \typespaces \to \feasible$ is \emph{weakly monotone} if for each agent~$i$, any $\typei, \typei' \in \typespacei$ and $\typesmi \in \typespacesmi$,
	\begin{align*}
		\Ex{\langle \typei - \typei', \alloci(\typei,\typesmi)  - \alloci(\typei', \typesmi) \rangle } \geq 0. 
	\end{align*}
\end{definition}
\begin{proposition}[see, e.g., \citeauthor{SY05} \citeyear{SY05}]
An allocation rule is implementable only if it is weakly monotone.
\end{proposition}

In fact, \citet{SY05} showed that, if $\typespaces$ is convex, then weak monotonicity is also a sufficient condition for DSIC implementability.

\paragraph{Revenue.}
A mechanism $(\allocs, \pays)$ is ex post \emph{individually rational} (IR)
if for each agent~$i$ and for every $\types \in \typespaces$, $\langle \typei, \alloc_{i}(\types) \rangle - \payi(\types) \geq 0$.

Given a distribution~$D$ on $\typespaces$ and a mechanism $(\allocs, \pays)$ that is DSIC and ex post IR, the \emph{expected revenue} of the mechanism is $\Ex[\types \sim D]{\sum_i \payi(\types)}$.  The \emph{optimal} revenue is the maximum expected revenue achievable among all DSIC, ex post IR mechanisms.

\section{DSIC Convex Extensions}
\label{sec:map}
\label{sec:mapping}
Before presenting our main result on the impossibility to extend DSIC allocation rules, we first complement \citet{DFK15}'s result in single-dimensional setting with a simple observation in multi-dimensional preference settings: whenever the feasible allocation space is downward closed, any DSIC allocation rule on a type space can be extended to its convex hull by another DSIC allocation rule.

\begin{theorem}
\label{thm:downward}
If the set of feasible allocations $\feasible$ is downward closed, for any DSIC allocation rule $\allocs$ on a type space~$\typespaces$, there is a DSIC extension~$\extallocs$ of~$\allocs$ on $\convhull(\typespaces)$.  If $\allocs$ is implemented with a payment rule~$\pays$, $\extallocs$ can be implemented by an extension $\extpays$ of~$\pays$.  If $\pays$ is individually rational on~$\typespaces$, so is $\extpays$ on $\convhull(\typespaces)$.
\end{theorem}

If we do not require the statement about individual rationality, extensibility is guaranteed by an even weaker condition, which we call \emph{single swap feasible}.

\begin{definition}[Single swap feasible] A feasible allocation set $\feasible$ is \emph{single swap feasible} (SSF) if for every agent~$i$ there exists an allocation $\swapfeasible(i) \in \feasible$ such that  for any $\allocs' \in \feasible$, $(\alloci',\swapfeasible_{-i}(i)) \in \feasible$. 
\end{definition}
Intuitively, 
$\swapfeasible(i)$ is a feasible allocation vector such that if we replace the $i^{\text{th}}$ element of this vector with the $i^{\text{th}}$ element from any other feasible allocation the resulting allocation is still feasible.  
If $\feasible$ is a product space or is downward closed, it must be SSF.  
\footnote{To see that a downward closed $\feasible$ is SSF, observe that we can let $\swapfeasible(i)$ be the all zero vector for each agent~$i$.}
\begin{theorem}
\label{thm:swap}
If the set of feasible allocations $\feasible$ is SSF, 
 any DSIC allocation rule $\allocs$ on a type space~$\typespaces$, there is a DSIC extension~$\extallocs$ of~$\allocs$ on $\convhull(\typespaces)$.
\end{theorem}

The proofs for both \Cref{thm:downward} and \Cref{thm:swap} can be found in the supplementary materials.
The main result of this paper is that without this condition a DSIC extension may not exist.

\begin{theorem}
    \label{thm:rand}
There is a two agent type space~$\randtypespaces$ with a DSIC allocation rule~$\randallocs$, such that $\randallocs$ cannot be extended by a DSIC allocation rule to $\convhull(\randtypespaces)$.
\end{theorem}

We prove the theorem in two steps.  We first present a setting with three-dimensional preferences for which we show the non-existence of \emph{deterministic} extensions.  We then build on the construction, lifting it to a higher dimension, where we strengthen the argument and show the non-existence of extensions that even allow randomization.

\subsection{Non-existence of deterministic extensions}

We first present type space $\dettypespaces = \dettypespacei[1] \times \dettypespacei[2]$ and the allocation rule $\detallocs$, and then show that $\detallocs$ is DSIC and yet cannot be extended by any deterministic DSIC allocation rule on $\convhull(\dettypespaces)$.

The two agents have identical type spaces: for $i = 1, 2$, $\dettypespacei = \dettypespace \coloneqq \{A = [1, 0, 0], B = [0, 1, 0], C = [0, 0, 1], D = [\tfrac 1 3, \tfrac 1 3, \tfrac 1 3]\}$. 
A visual representation of this typespace and its convex hull can be found in figure 1 in the supplementary materials.

The allocation rule $\detallocs$ is also symmetric, in the sense that 
$\detalloci[1](\typei[1], \typei[2]) = \detalloci[2](\typei[2], \typei[1])$ for any $\typei[1], \typei[2] \in \dettypespace$.  
We summarize $\detalloci[1]$ with the diagram below.  
The rows are indexed by agent~$1$'s own type~$V_1$, and the columns by agent~$2$'s type~$V_2$:
\begin{align*}
\begin{blockarray}{c cccc}
 & A & B & C & D \\
\begin{block}{c(cccc)}
A & [1,1,0] & [2,0,2] & [3,0,3] & [4,0,4] \\ 
B & [0,1,1] & [2,2,0] & [3,3,0] & [4,4,0] \\ 
C & [1,0,1] & [0,2,2] & [0,3,3] & [0,4,4] \\ 
D & [0,1,1] & [2,2,0] & [3,3,0] & [4,4,0]  \\
\end{block}
\end{blockarray} 
\end{align*}

The set of all feasible allocations is then $\feasible = \{(\detalloci[1](V_1, V_2), \detalloci[2](V_1, V_2) )\}_{V_1, V_2 \in \dettypespace}$.  
We hasten to point out that $\feasible$ is \emph{not} the product between the two agents' respective set of feasible allocations.  
For example, $[1, 1, 0, 1, 1, 0]$ is in~$\feasible$ as it is $\detallocs(A, A)$, but $[1, 1, 0, 0, 1, 1]$ is not.
This is important for the proof.

\begin{proposition}
$\detallocs$ is DSIC implementable.
\end{proposition}

\begin{proof}
Let the payment be \num 0 for both agents and all type profiles.  
As the allocation and payment rules are both symmetric, consider either agent~$i$.  If $\typei[-i] = A$, the maximum value agent~$i$ could get, when her type is $A$, $B$, or~$C$, is \num 1, attained with truthful bidding.  For $\typei = D$, the four allocations all give the same value $\tfrac 2 3$.  Similar arguments hold when $\typei[-i]$ is $B$, $C$ or~$D$.
\end{proof}

\begin{theorem}
    \label{thm:det}
    There exists no deterministic DSIC extension of~$\detallocs$.
\end{theorem}

Before proving Theorem~\ref{thm:det}, we make several preparatory observations.

A key difficulty with multi-dimensional preferences is the the lack of a payment identity \`a la Myerson \cite{Myerson81}.  
In order to argue that any extension of an allocation rule is not DSIC, one has to either check the many cyclic (or weak) monotonicity conditions, or show that no payment rule can support the extension in a DSIC mechanism.
We designed $\dettypespaces$ and~$\detallocs$ carefully so that the allocations ``lock'' the payment rules.

\begin{lemma}
\label{lem:equal-pay}
For any allocation rule $\extallocs$ that is an extension of $\detallocs$, if $\extallocs$ can be implemented by a DSIC mechanism with payment rule~$\pays$, then for any $V \in \{A, B, C, D\}$, $\payi[1](A, V) = \payi[1](B, V) = \payi[1](C, V) = \payi[1](D, V)$,
and
$\payi[2](V, A) = \payi[2](V, B) = \payi[2](V, C) = \payi[2](V, D)$.
\end{lemma}

\begin{proof}
We prove the lemma for agent~$1$, and the statement for agent~$2$ follows by symmetry.
Note that when $V=A$
\begin{align*}
    \langle A , \detalloci[1](A,V)\rangle = \langle A , \detalloci[1](C,V)\rangle; \\
    \langle C , \detalloci[1](C,V)\rangle = \langle C , \detalloci[1](B,V)\rangle; \\
    \langle B , \detalloci[1](B,V)\rangle = \langle B , \detalloci[1](A,V)\rangle.
    \end{align*}
By DSIC, for any $V \in \dettypespace$, we have
\begin{align*}
    \langle A , \detalloci[1](A,V)\rangle & -\payi[1](A,V)  \\
     & \geq  \langle A , \detalloci[1](C,V)\rangle -\payi[1](C,V); \\
    \langle B , \detalloci[1](B,V)\rangle & -\payi[1](B,V) \\
    & \geq \langle B , \detalloci[1](A,V)\rangle -\payi[1](A,V);
     \\
     \langle C, \detalloci[1](C,V)\rangle & -\payi[1](C,V) \\
     & \geq \langle C , \detalloci[1](B,V)\rangle -\payi[1](B,V).
\end{align*}

Therefore
     $\payi[1](A,V) \leq  \payi[1](C,V) \leq \payi[1](B, V) \leq \payi[1](A, V)$.
Hence all inequalities are tight and we have 
$\payi[1](C,V) = \payi[1](A,V) =  \payi[1](B,V)$.

Similarly when $V\neq A$ we have,
\begin{align*}
    \langle A , \detalloci[1](A,V)\rangle = \langle A , \detalloci[1](B,V)\rangle; \\
    \langle B , \detalloci[1](B,V)\rangle = \langle B , \detalloci[1](C,V)\rangle; \\
    \langle C , \detalloci[1](C,V)\rangle = \langle C , \detalloci[1](A,V)\rangle.
    \end{align*}
    
Again by using the corresponding DSIC constraints we have  $\payi[1](A,V) \leq  \payi[1](B,V) \leq \payi[1](C, V) \leq \payi[1](A, V)$.
Hence all inequalities are tight and we have 
$\payi[1](C,V) = \payi[1](A,V) =  \payi[1](B,V)$.

For type~$D$'s payment, note that $\detalloci[1](B, V)= \detalloci[1](D, V)$. 
If $\payi[1](B, V) \neq \payi[1](D, V)$, one of $B$ and~$D$ must be incentivized to misreport the other type.  Therefore $\payi[1](B, V) = \payi[1](D, V)$.
\end{proof}
In figure 2 in the supplementary materials we give a $3$ dimensional visualization of the agent's value for each allocation which gives intuition into the proof of \cref{lem:equal-pay}.

\begin{lemma}
\label{lem:fixed-menu}
If $\tilde{x}$ is a deterministic DSIC extension of~$\detallocs$, then for $\typei[1] = $ $\frac{1}{3}A + \frac{1}{3}B + \frac{1}{3}D$, and any $V_2 \in \dettypespace$, $\extalloci[1](\typei[1], V_2) \in $ $\{\extalloci[1](A, V_2), \extalloci[1](B, V_2),$ $\extalloci[1](C, V_2),$ $ \extalloci[1](D, V_2)\}$.
\end{lemma}

\begin{proof}
For the sake of contradiction, assume $\extallocs$ is a DSIC extension of~$\detallocs$, implementable by payment rule~$\pays$, and for $\typei[1]$
and $V_2 \in \dettypespace$, $\extalloci[1](\typei[1], V_2) = \detalloci[1](V_1, V_2')$ for some $(V_1, V_2') \in \dettypespaces$ and $V_2' \neq V_2$.

For any $V_2$, one of $\detalloci[1](A, V_2)$ and $\detalloci[1](B, V_2)$ gives an equal positive value to both $A$ and $B$.  Let $V_1^*$ be the type that induces this equally valued allocation. (For example, if $V_2 = B$, then $\detalloci[1](A, B) = [2, 0, 2]$ and $\detalloci[1](B, B) = [2, 2, 0]$. Both $A$ and $B$ have the same value for $\detalloci[1](B, B)$ and so type~$B$ would be $V_1^*$.)

Observe that, for the allocation $\detalloci[1](V_1, V_2')$, $V_1$ has positive value $\langle V_1, \detalloci[1](V_1, V_2') \rangle$, and no other type has higher value for it.
Therefore, type~$\typei[1]$ has value at most $\langle \frac{2}{3} V_1 + \frac{1}{3} D, \detalloci[1](V_1, V_2') \rangle$ for the allocation.
In order for $\typei[1]$ to have no incentive to misreport~$V_1^*$, we must have
\begin{multline}
\langle \frac{2}{3} V_1 + \frac{1}{3} D, \detalloci[1](V_1, V_2') \rangle - \payi[1](\typei[1], V_2) 
\\ 
\geq \langle \frac{2}{3} V_1^* + \frac{1}{3} D, \detalloci[1](V_1^*, V_2) \rangle - \payi[1](V_1^*, V_2) 
\label{eq:t1-V1}
\end{multline}

On the other hand, in order for type~$V_1$ not to have incentive for deviating to~$\typei[1]$, we have
\begin{multline}
    \langle V_1, \detalloci[1](V_1, V_2) \rangle - \payi[1](V_1, V_2) 
    \\
    = \langle V_1^*, \detalloci[1](V_1^*, V_2) \rangle - \payi[1](V_1^*, V_2) 
    \\
    \geq \langle V_1, \detalloci[1](V_1, V_2') \rangle - \payi[1](\typei[1], V_2);
    \label{eq:V1-t1}
    \end{multline}
where for the equality we used the fact that the value obtained by reporting truthfully is the same for every type in $\{A, B, C\}$ given a fixed type of the opponent, and that $\payi[1](V_1, V_2) = \payi[1](V_1^{*}, V_2)$ by Lemma~\ref{lem:equal-pay}.  

Similarly, in order for type~$D$ not to have incentive for deviating to~$\typei[1]$, we have
    \begin{multline}
    \langle D, \detalloci[1](D, V_2) \rangle - \payi[1](D, V_2) 
    \\
    = \langle D, \detalloci[1](V_1^*, V_2) \rangle - \payi[1](V_1^*, V_2) 
    \\
    \geq \langle D, \detalloci[1](V_1, V_2') \rangle - \payi[1](\typei[1], V_2),
    \label{eq:D-t1}
\end{multline}
where for the equality we used the fact that type~$D$ has the same value for all allocations given a fixed type of the opponent, and that $\payi[1](D, V_2) = \payi[1](V_1, V_2)$ by Lemma~\ref{lem:equal-pay}.
Crucially, \eqref{eq:V1-t1} and~\eqref{eq:D-t1} cannot both be tight, because by construction, for any $V_2 \neq V_2'$, 
\begin{multline*}
\langle V_1, \detalloci[1](V_1, V_2) - \detalloci[1](V_1, V_2') \rangle \\
= \frac{3}{2} \langle D, \detalloci[1](D, V_2) - \detalloci[1](V_1, V_2') \rangle \neq 0.  
\end{multline*}
Therefore, $\frac{2}{3} \cdot$ \eqref{eq:V1-t1} $+ \frac{1}{3} \cdot$ \eqref{eq:D-t1} gives
\begin{multline*}
\langle \frac{2}{3} V_1^* + \frac{1}{3} D, \detalloci[1](V_1^*, V_2) \rangle - \payi[1](V_1^*, V_2)
\\ 
> \langle \frac{2}{3} V_1 + \frac{1}{3} D, \detalloci[1](V_1, V_2') \rangle - \payi[1](\typei[1], V_2),
\end{multline*}
which contradicts \eqref{eq:t1-V1}.
\end{proof}
By the same reasoning as for Lemma~\ref{lem:equal-pay}, the following lemma follows from Lemma~\ref{lem:fixed-menu}.

\begin{lemma}
\label{lem:general-equal-pay}
If $\extallocs$ is a deterministic DSIC extension of~$\detallocs$, implementable by payment rule $\pays$, then for any $\typei[1]$ in the interior of $\convhull(\dettypespace)$ and any $V_2 \in \dettypespace$, $\payi[1](\typei[1], V_2) = \payi[1](A, V_2)$.
\end{lemma}

We are now ready to prove Theorem~\ref{thm:det}.
\begin{proof}[Proof of Theorem~\ref{thm:det}]
Suppose $\extallocs$ is a deterministic DSIC extension of~$\detallocs$.
We show a contradiction by showing that $\extallocs$ must violate weak monotonicity.

Consider $\typei[1] = \tfrac 1 3 A + \tfrac 1 3 B + \tfrac 1 3 D$ and when agent~$2$'s type is~$A$.  
Since $\typei[1]$ could report any type in $\dettypespace$, she has as options $\detalloci[1](A, A), \detalloci[1](B, A), \detalloci[1](C, A)$ and $\detalloci[1](D, A) \}$, all at the same price by Lemma~\ref{lem:general-equal-pay}.
By Lemma~\ref{lem:fixed-menu}, these are also all the allocations she could possibly get.

Since $[1, 1, 0]$ is the only allocation for which both types $A$ and~$B$ have positive value, it is $\typei[1]$'s preferred allocation, i.e., $\extalloci[1](\typei[1], A)$ must be $[1, 1, 0]$.  
This in turn implies $\extalloci[2](\typei[1], A) = [1, 1, 0]$.  
(Recall that $\feasible$ is not a product space, and the only allocation in which agent~$1$ gets $[1, 1, 0]$ is $\detallocs(A, A) = [1, 1, 0, 1, 1, 0]$.) 

Similarly, one can show $\extalloci[1](\typei[1], D) = [4, 4, 0]$, which implies $\extalloci[2](\typei[1], D) = [2, 2, 0]$ or $[4, 4, 0]$.  But in either case, weak monotonicity is violated for agent~$2$'s types $A$ and~$D$.  For example, if $\extalloci[2](\typei[1], D) = [2, 2, 0]$, we have
\begin{align*}
    \langle A, [1, 1, 0] \rangle + \langle D, [2, 2, 0] \rangle < \langle A, [2, 2, 0] \rangle + \langle D, [1, 1, 0] \rangle.
\end{align*}
Therefore, no deterministic DSIC extension of~$\detallocs$ is possible.
\end{proof}

\subsection{Non-existence of randomized extensions}

We need a more convoluted construction and a more careful argument to prove the impossibility of extensions that are possibly randomized.  
We build on $\dettypespaces$ and~$\detallocs$ to construct $\randtypespaces$ and~$\randallocs$ and prove Theorem~\ref{thm:rand}.

We first raise types in~$\dettypespace$ to a space of seven dimensions.  Define $A' = [1, 0, 0, 0, 0, 0, 0]$, $B' = [0, 1, 0, 0, 0, 0, 0]$, $C' = [0, 0, 1, 0, 0, 0, 0]$ and $D' = \frac 1 3 (A' + B' + C')$.
For ease of notation, we define a mapping $\det: \{A', B', C', D'\} \to \dettypespace$, with $\det(A') = A, \det(B') = B, \det(C') = C$ and $\det(D') = D$.  
We also introduce four new types, $E' = [0, 0, 0, 1, 0, 0, 0]$, 
$F'=[0,0,0,0,1,0,0]$, $G'=[0,0,0,0,0,1,0]$ and $H'=[0,0,0,0,0,0,1]$.  Define $\randtypespace = \{A', B', C', D', E', F', G'\}$,  and $\randtypespaces = \randtypespace \times \randtypespace$.

We now define~$\randallocs$, which is again symmetric, in the sense that $\randalloci[1](V_1, V_2) = \randalloci[2](V_2, V_1)$ for every $V_1, V_2 \in \randtypespace$.  We therefore only describe $\randalloci[1]$.  When both agents report types in $\{A', B', C', D'\}$, the first three coordinates of each agent's allocation are given by $\detallocs$ when fed by the corresponding types in~$\dettypespaces$, and the remaining coordinates are filled in according to the opponent's report.  More specifically,
\begin{multline}
\forall V_1 \in \{A',B',C',D'\}, \\
\randalloci[1](V_1, A')=[\detalloci[1](\det(V_1), A), 0, 100, 100,100],\\
\randalloci[1](V_1, B')=[\detalloci[1](\det(V_1), B),100 , 0, 100, 100],\\
\randalloci[1](V_1, C')=[\detalloci[1](\det(V_1), C),100, 100, 0, 100], \\
\randalloci[1](V_1, D')=[\detalloci[1](\det(V_1), D), 100, 100, 100,0]
\end{multline}

For the other types, we have
\begin{multline*}
    \forall V_1 \in \{E', F', G', H'\}, \forall V_2 \in \{A', B', C', D'\}, \\
    \randalloci[1](V_1, V_2) = \randalloci[1](C', V_2).\\
    \forall V_2 \in \{E', F', G', H'\}, \forall V_1 \in \randtypespace, \\
    \randalloci[1](V_1, V_2) = [0, 0, 0, 100, 100, 100, 100].
\end{multline*}

Note that $\randallocs$ itself is deterministic.
The difficulty we need to overcome in this section is that the \emph{extension} of~$\randallocs$ may be randomized, and we must show that any extension to~$\convhull(\randtypespaces)$ cannot be DSIC.

The set of feasible allocations is $\feasible $ $= $ $\{\randallocs(V_1, V_2)\}_{V_1, V_2 \in \randtypespace}$.
Again we emphasize that the set of feasible allocations is \emph{not} a product space.  

We first show that a subset of the payments are still ``locked'' as they were in the deterministic setting. 
\begin{lemma}
\label{lem:randequal-pay}
For any allocation rule $\extallocs$ that is an extension of $\randallocs$, if $\extallocs$ can be implemented by a DSIC mechanism with payment rule~$\pays$, then for any $V \in \{A', B', C', D'\}$ and any $V',V'' \in \randtypespace$, $\payi[1](V', V) = \payi[1](V'', V)$,
and
$\payi[2](V, V') = \payi[2](V, V'')$.
\end{lemma}
\begin{proof}
The proof for types in $\{A',B',C',D'\}$ follows the same steps as \Cref{lem:equal-pay}. For any type in $\{E',F',G',H'\}$ note that they receive the same allocation as type $C'$ and therefore if they are charged a payment that is different from the payment $C'$ is charged than either $C'$ or that type would be incentivized to deviate.
\end{proof}

In order to have a DSIC convex extension in this setting we must satisfy a condition similar to \cref{lem:fixed-menu} with the higher dimensional versions of the types from $\dettypespaces$.

\begin{lemma}
If $\extallocs$ is a DSIC extension of~$\randallocs$, then for $\typei[1] = \frac{1}{3}A' + \frac{1}{3}B' + \frac{1}{3}D'$,
and any $V_2 \in \{A',B',C',D'\}$,
$\extalloci[1](\typei[1], V_2)$ is supported on 
$\{\extalloci[1](A', V_2), \extalloci[1](B', V_2), \extalloci[1](C', V_2), \extalloci[1](D', V_2)\}$.

\label{lma:invariant-random}
\end{lemma}
\begin{proof}
For the sake of contradiction, assume $\extallocs$ is a DSIC extension of~$\randallocs$, implementable by payment rule~$\pays$, and for some $V_2 \in \{A',B',C',D'\}$,

there exist a pair of types $\vertypei[1], \vertypei[2]'$ such that $\Prx{\extallocs(\typei[1], \vertypei[2]) = \randallocs(\vertypei[1], \vertypei[2]')} > 0$ for a $\vertypei[2]'\neq\vertypei[2]$. 
Let $\fixedmenuset$ denote the set of allocations $\{\extalloci[1](A', V_2), \extalloci[1](B', V_2), \extalloci[1](C', V_2), \extalloci[1](D', V_2)\}$.

Let $\alpha = \Prx{\extalloci[1](\typei[1], \vertypei[2]) \in \fixedmenuset}$ and 
$\beta = \Prx{\extalloci[1](\typei[1], \vertypei[2]) \not \in \fixedmenuset}.$  
Then by assumption, $\beta > 0$.

Also let $\extalloci[1]^{\alpha}(\typei[1]) = \Ex{\extalloci[1](\typei[1],\vertypei[2])  \mid \extalloci[1](\typei[1],\vertypei[2]) \in \fixedmenuset}$ 
and $\extalloci[1]^{\beta}(\typei[1]) = \Ex{\extalloci[1](\typei[1],\vertypei[2])  \mid \extalloci[1](\typei[1],\vertypei[2]) \not\in \fixedmenuset}$.
Then the expected value for agent~$1$ truthfully reporting~$\typei[1]$ is $\langle \typei[1], \alpha \extalloci[1]^\alpha (\typei[1])+ \beta \extalloci[1]^\beta(\typei[1]) \rangle$.

In order for $\extallocs$ to be DSIC, it must provide type $\typei[1]$ with at least as much utility as it could receive from deviating to any type in $\randtypespace$.  

Let $\vertypei[\text{max}] \in \argmax_{\vertypei[1]'\in \randtypespace} \langle \typei[1],\randalloci[1](\vertypei[1]',\vertypei[2])) \rangle$ then $\extallocs$ must satisfy the following DSIC constraint,
\begin{align}
    \alpha &\langle \typei[1],  \extalloci[1]^{\alpha}(\typei[1]) \rangle + \beta \langle \typei[1], \extalloci[1]^{\beta}(\typei[1])\rangle - \payi[1](\typei[1],\vertypei[2]) \geq \nonumber \\ &\langle \typei[1],\randalloci[1](\vertypei[\text{max}],\vertypei[2])) \rangle - \payi[1](\vertypei[\text{max}],\vertypei[2]). \label{eq:randIC-contr}
\end{align}
 We use this constraint to place an upper bound on $\payi[1](\typei[1],\vertypei[2])$.

By definition of~$V_{\text{max}}$, 
\begin{align*}
    \langle \typei[1],\randalloci[1](\vertypei[\text{max}],\vertypei[2])) \rangle \geq \langle \typei[1], \extalloci[1]^{\alpha}(\typei[1]) \rangle.
\end{align*}
 Subtracting $\langle \typei[1],\randalloci[1](\vertypei[\text{max}],\vertypei[2])) \rangle$ from both sides of \cref{eq:randIC-contr} and utilizing the fact that $\alpha+\beta=1$ we have
\begin{align*}
    \beta \langle \typei[1], \extalloci[1]^{\beta}(\typei[1]) \rangle - \beta \langle \typei[1],\randalloci[1](\vertypei[\text{max}],\vertypei[2]) \rangle \\ \geq \payi[1](\typei[1],\vertypei[2]) - \payi[1](\vertypei[\text{max}],\vertypei[2]).
\end{align*}
Now using the fact that $\typei[1]$'s maximum difference in value for any two allocations in~$\feasible$ is upper bounded by~3, we get
\begin{align}
    3\beta+\payi[1](\vertypei[\text{max}],\vertypei[2]) &\geq \payi[1](\typei[1],\vertypei[2]). \label{eq:pay-bound}
\end{align}
 
By construction, for $\vertypei[2] \in \{A', B', C', D'\}$, there exists a type $\vertypei[\text{not}] \in \{E',F',G',H'\}$ which has value $0$ for every allocation in $\fixedmenuset$. (For example if $V_2=A'$ then all of the allocations in $\{\extalloci[1](A',A'), \extalloci[1](B',A'), \extalloci[1](C', A'), \extalloci[1](D', A')\}$ have a $0$ in the $4^{\text{th}}$ coordinate and in this case $V_{\text{not}}=E'$). Now by \Cref{lem:randequal-pay}, $V_{\text{not}}$ receives utility  $-\payi[1](V_{\text{not}}, \vertypei[2]) = -\payi[1](\vertypei[\text{max}],\vertypei[2])$ when reporting truthfully against opponent type $\vertypei[2]$.

Therefore, in order to keep $\vertypei[\text{not}]$ from deviating to type $\typei[1]$, we must satisfy the following DSIC constraint,
\begin{align}
    -\payi[1](\vertypei[max],\vertypei[2]) \geq & \alpha \langle \vertypei[\text{not}],  \extalloci[1]^{\alpha}(\typei[1]) \rangle \nonumber + \beta \langle \vertypei[\text{not}], \extalloci[1]^{\beta}(\typei[1])\rangle \\
    & - \payi[1](\typei[1],\vertypei[2]). 
    \label{eq:tempDSIC}
\end{align}
Since $\extalloci[1]^{\alpha}(\typei[1])$ is the expected allocation conditioned on the resulting allocation being an element of $\fixedmenuset$, and $V_{\text{not}}$ has zero value for any allocation in $\fixedmenuset$, 
we have $\langle V_{\text{not}}, \extalloci[1]^\alpha \rangle = 0$.

Therefore, 

\begin{align}
    \payi[1](\typei[1],\vertypei[2])  &\geq \beta \langle \vertypei[\text{not}], \extalloci[1]^{\beta}(\typei[1])\rangle + \payi[1](\vertypei[\text{max}],\vertypei[2]). \label{eq:pay-bound2}
\end{align}

Notice that for every allocation outside of $\fixedmenuset$ the type $\vertypei[\text{not}]$ has value exactly $100$ and therefore $\langle \vertypei[\text{not}], \alloc_{\beta}(\type)\rangle = 100$. We now use this fact along with \cref{eq:pay-bound} and \cref{eq:pay-bound2} and derive the following contradiction,
\begin{align*}
    3\beta+\payi[1](\vertypei[\text{max}],\vertypei[2])- \payi[1](\vertypei[\text{max}],\vertypei[2]) &\geq \beta \langle \vertypei[\text{not}], \alloc_{\beta}\rangle, \\
    \Rightarrow \quad 3\beta &\geq 100 \beta.
\end{align*}
Since by assumption $\beta > 0$, this is a contradiction.  

\end{proof}

We are now ready to prove Theorem~\ref{thm:rand}.
\begin{proof}[Proof of Theorem~\ref{thm:rand}]

Consider $\typei[1] = \tfrac 1 3 A' + \tfrac 1 3 B' + \tfrac 1 3 D'$  and assume $\Prx{\extalloci[1](\typei[1],A') = \randalloci[1](A', A')} <1$. 

We know from \Cref{lma:invariant-random} that $\extalloci[1](\typei[1], A')$ is supported on $\{\extalloci[1](A',A'), \extalloci[1](B',A'), \extalloci[1](C', A'), \extalloci[1](D', A')\}$. Since $\randalloci[1](A',A')$ is the only allocation for which both $A'$ and $B'$ have positive value, $\typei[1]$ satisfies
\begin{align}
    \forall \vertype' \neq A', \quad \langle \typei[1], \randalloci[1](\vertype',A') \rangle < \langle \typei[1], \randalloci[1](A',A') \rangle.\label{eq:strictval}
\end{align}
But in order for $\typei[1]$ to have no incentive to report~$A'$ when agent~$2$'s type is~$A'$, we have
\begin{align*}
    \Ex{\langle \typei[1], \extalloci[1](\typei[1],A') \rangle} & - \payi[1](\typei[1],A') \geq \\
    & \langle \typei[1], \randalloci[1](A', A') \rangle - \payi[1](A', A').
\end{align*}
Using eq.~\eqref{eq:strictval}, we have
\begin{align}
    \payi[1](A', A') &> \payi[1](\typei[1],A')\label{eq:strictpay}.
\end{align}
We now show that, similar to \Cref{lem:general-equal-pay}, a strictly lower payment for $\typei[1]$ as in \Cref{eq:strictpay} would violate DSIC constraints.
 The DSIC constraint that keeps type $D'$ from deviating to type $\typei[1]$ when the opponents type is $A'$ is
\begin{align}
        \langle D', \randalloci[1](D', A') \rangle - \payi[1](D', A') \nonumber \\ \geq \langle D' , \extalloci[1](\typei[1],A') \rangle -\payi[1](\typei[1],A') \label{eq:Dconstraint}.
\end{align}
Since $\extalloci[1](\typei[1],A')$ is supported on $\{\extalloci[1](A',A'),$ $ \extalloci[1](B',A'),$ $ \extalloci[1](C', A'),$ $\extalloci[1](D', A')\}$ and type $D'$ values all of these allocations equally we have
\begin{align*}
    \langle D', \extalloci[1](\typei[1],A') \rangle = \langle D', \randalloci[1](D', A') \rangle,
\end{align*}
and by \Cref{lem:randequal-pay} we have $\payi[1](D',A') = \payi[1](A',A')$. Using these facts we write \Cref{eq:Dconstraint} as
\begin{align*}
     \payi[1](\typei[1],A') &\geq \payi[1](A',A').
\end{align*}
This contradicts \Cref{eq:strictpay}.

Therefore, $\extalloci[1](\typei[1], A')$ must be deterministically $\randalloci[1](A', A')$.   
But the only allocation in~$\feasible$ with agent~$1$'s allocation being this is $\randallocs(A', A')$, 
and hence $\extalloci[2](\typei[1], A')$ is deterministically $[1, 1, 0,0,100,100,100]$.
Using the same argument, one can show $\extalloci[1](\typei[1], D') = [4, 4, 0,100,100,100,0]$ with probability 1, which implies $\extalloci[2](\typei[1], D')$ is supported on $\{[2, 2, 0,100,0,100,100],[4, 4, 0,100,100,100,0]\}$. For any allocation with this support, weak monotonicity is violated for agent~$2$'s types $A'$ and~$D'$.
\end{proof}
\section{Revenue Gap}
\label{sec:revenue}
In this section we explore the revenue gap between mechanisms that satisfy DSIC and IR on $\supp(D)$ and mechanisms that satisfy DSIC and IR on $\convhull(\supp(D))$. Any type that is in $\convhull(\supp(D))$ and not $\supp(D)$ occurs with probability zero and does not itself contribute to the revenue. It is therefore not obvious that enforcing constraints on these additional types will impact revenue. In fact, \Cref{thm:downward} shows that if the feasible set is downward closed there is no gap in revenue between these two typespaces. However, for arbitrary feasibility sets, a gap can arise because not every mechanism which is DSIC and IR can be extended to $\convhull(\supp(D))$ without changing the payment rules on the support. We show that for some feasibility sets and distributions the gap in revenue between the optimal DSIC and IR mechanism on $\supp(D)$ and the optimal mechanism on $\convhull(\supp(D))$ is unbounded.

 For a feasible set of allocations $\feasible$ and a distribution $\mathbf{D}$ we, let 
 $\OPT(\feasible,\mathbf{D})$ denote the optimal revenue extractable by a DSIC and IR mechanism on $\supp(D)$, and let
 $\widetilde{\OPT}(\feasible,\mathbf{D})$ denote the maximum revenue extractable by a mechanism that is DSIC and IR on $\Conv(\supp(D))$.

 We focus on the single agent case and show that the ratio between $\OPT(\feasible, D)$ and $\widetilde{\OPT}(\feasible, D)$ is unbounded. 
\begin{theorem}
	\label{thm:revenue}
$\forall \alpha$, there exists a feasible set $\feasible$ and distribution $\mathbf{D}$ such that   $\OPT(\feasible,\mathbf{D}) \geq \alpha \widetilde{\OPT}(\feasible, \mathbf{D})$. 
\end{theorem}
\begin{proof}
Let $\allones$ be the "all-ones" vector of length k and define the set 
\begin{align*}
\singleonevec = \{ \type \in \mathbb{R}^k  \mid \exists j \text{ s. t. } \typei[j] =1 \\ \text{ and } \forall h \neq j, \: \typei[h]=0\}.
\end{align*}
Intuitively, $\singleonevec$ is the set of all vectors of length k with exactly one coordinate being $1$ and the rest of the coordinates being zero.
We define the support, 
$\supp(D) = \singleonevec \cup \vec{\mathbf{1}}$
and the set of feasible allocations $\feasible = \singleonevec$.

We can show that $\OPT(\feasible, \mathbf{D})=1$ for any distribution $\mathbf{D}$ that shares this support. This is obtained by giving each $\type \in \singleonevec$ an allocation $\alloc$ equal to their type i.e. $\type = \alloc(\type)$ and giving type $\type = \allones$ any allocation in the feasible set. We can now charge a payment of $1$ to every type.

Without loss of generality assume $\alpha>1$. Now for any $\alpha$ define $\epsilon< \frac{1}{\alpha}$ and consider the distribution where 
$Pr[\type= \vec{\mathbf{1}}] = 1- \epsilon$ and $Pr[\type \neq \vec{\mathbf{1}}] = \epsilon$. One of the types contained in $\convhull(\supp(D))$ is $\type^{k} = [\frac{1}{k}, \frac{1}{k} \ldots \frac{1}{k}]$. $\type^{k}$ has a value of $\frac{1}{k}$ for every feasible allocation and every lottery of existing allocations, therefore the maximum payment we can charge without violating IR constraints is $\frac{1}{k}$. However, the all ones type has value $1$ for every feasible allocation and can now gain utility of $\frac{k-1}{k}$ by deviating to type $\type^k$. Therefore to maintain incentive compatibility we must lower the price we charge the all ones type to $\frac{1}{k}$ (as well as the payment charged to one of the types $\singleonevec$). Therefore the revenue generated on the convex hull if we take $k$ to be large is, $\widetilde{\OPT}(\feasible, \mathbf{D})\leq \lim_{k\rightarrow \infty} (1-\epsilon)(\frac{1}{k})+\epsilon \leq \frac{1}{\alpha}$.
\end{proof}
\section{Conclusion and Future Work}
\label{sec:disc}
    We presented a DSIC mechanism on a discrete type space which cannot be implemented if DSIC is required on the convex hull of its type space.
    Extension from discrete to convex domains is a common step in automated mechanism design, and it is unfortunate that, as an implication of our result, there exists no general procedure for extending an arbitrary DSIC mechanism to the convex hull of its domain.
    However, as we showed, such extensions are possible in specific settings, with special type spaces (e.g.\@ single dimensional) or special feasible set of allocations (e.g.\@ downward closed or, more generally, single swap feasible).
    Other sufficient conditions for extensibility should be useful to have.
    
    Another possibility we leave open is conditions on an allocation rule (rather than the type space or the set of feasible allocations) that would guarantee its extensibility.  
    We know that an allocation rule is DSIC implementable if and only if it is cyclic monotone.
    Is there a condition stronger than cyclic monotonicity which guarantees both implementability and extensibility?
    
    The mechanism we constructed for the impossibility result is intricate.  A third question we leave open is whether one may find a succinct property that is \emph{necessary} for extensibility; in other words, is there a natural condition whose violation must result in inextensible mechanisms?

\bibliographystyle{aaai}
\bibliography{ref}

\appendix
\section{Supplementary Material}
\label{sec:supp}
\subsection{Proof of \Cref{thm:downward}}
\begin{proof}
For ease of notation let $\zeroalloc$ represent the ``zero-allocation'' (i.e. an all-zeros vector of length $m$). 
We now give, for any DSIC and IR $(\allocs, \pays)$, 
an extended allocation rule $\extallocs$ and extended payment rule $\extpays$ that are DSIC and IR.

Whenever the reported type profile $\types$ contains only types from the original typespace $\typespaces$ (i.e. $\types \in \typespaces$) then by definition our extended allocation rule is $\extallocs(\types)=\allocs(\types)$ and our extended payment rule is $\extpays(\types)=\pays(\types)$.

If exactly one agent~$i$ has a type $\typei \in \convhull(\typespacei)\setminus \typespacei$ and  their opponents have types $\typesmi \in \typespacesmi$, then agent~$i$'s allocation is the same allocation as either the type $\typei^\text{max} = \argmax_{\typei'\in \typespacei} \langle \typei, \alloci(\typei',\typesmi) \rangle - \payi(\typei', \typesmi)\rangle$ or the ``zero-allocation'', whichever provides more utility. 
More formally,  if $\typei \in \convhull(\typespacei)\setminus \typespacei$ and $\typesmi \in \typespacesmi$, the extended allocation rule for $\typei$ is
\begin{align*}
  \extalloci(\typei,\typesmi) =\begin{cases}
  \alloci(\typei^\text{max},\typesmi), & 
  \!\begin{aligned}
  \text{if } \langle \typei, \alloci(\typei^{\text{max}},\typesmi) \rangle \\
  - \payi(\typei^{\text{max}}, \typesmi) \geq 0;
  \end{aligned}
  \\[10pt]
  \zeroalloc, & \text{otherwise;}
  \end{cases}
\end{align*}
and the extended payment rule is
\begin{align*}
  \extpayi(\typei,\typesmi) =\begin{cases}
  \payi(\typei^\text{max},\typesmi), & 
  \!\begin{aligned}
  \text{if } \langle \typei, \alloci(\typei^{\text{max}},\typesmi) \rangle \\
  - \payi(\typei^{\text{max}}, \typesmi) \geq 0;
  \end{aligned}
  \\[10pt]
  0, & \text{otherwise.}
  \end{cases}
\end{align*}

 All other agents $j \neq i$ are given the ``zero-allocation'' and charged a payment of $0$ (i.e. $\extalloci[j](\typei[j], \typesmi[j]) = \zeroalloc$ and $\extpayi[j](\typei[j], \typesmi[j]) = 0$). 

The allocation is feasible 
since  $\allocs(\typei^{\text{max}},\typesmi) \in \feasible$ 
and $\feasible$ is downward closed.


Finally if more than one agent has a type which is not in the original typespace then $\forall i$, $\extalloci(\typei, \typesmi) = \zeroalloc$ and $\extpayi(\typei, \typesmi) = 0$.

It remains to be shown that $\extallocs$ and $\extpays$ are DSIC and IR. To see that $\extallocs$ and $\extpays$ are IR observe that every ``new'' allocation and payment that are not zero are only chosen if they provide non-negative utility and since the ``zero-allocation'' with a payment of $0$ by definition provides $0$ utility, $\extallocs$ must satisfy IR constraints.

We now show that $\extallocs$ and $\extpays$ satisfy DSIC constraints. For any agent $i$ if they have an opponent $j$ who reports $\typei[j] \in \convhull(\typespacei[j])\setminus \typespacei[j]$ then agent $i$ is given the zero-allocation, regardless of their reported type. Therefore, we trivially satisfy DSIC constraints whenever an opponent type profile contains a ``new'' type. 

When $\typesmi \in \typespacesmi$ we must show that agent $i$ has no incentive to deviate, both when they have a ``old'' type and when they have an ``new'' type. 
For an ``old'' type  $\typei \in \typespacei$, deviating to some type $\typei' \in \convhull(\typespacei)\setminus \typespacei$ gives either the same utility as deviating to $\typei^\text{max} = \argmax_{\typei''\in \typespacei} \langle \typei', \alloci(\typei'',\typesmi) \rangle - \payi(\typei'', \typesmi)$ or utility~$0$.  
Since $(\allocs, \payments)$ is DSIC and IR, this deviation cannot provide higher utility than $\langle \typei,\extalloci(\typei,\typesmi) \rangle - \extpayi(\typei, \typesmi)$. 
Finally, for a ``new'' type $\typei \in \convhull(\typespacei)\setminus \typespacei$, by definition of~$\extallocs$, the agent by reporting truthfully is given the allocation and payment that maximizes utility over all possible deviations in $\typespacei$; by the logic above, this maximizes utility over all possible deviations in $\convhull(\typespacei)$ as well. Therefore, the agent has no incentive to deviate.
\end{proof}
\subsection{Proof of \Cref{thm:swap}}
\begin{proof}
This proof follows the same steps as the proof of \Cref{thm:downward}, with swap feasible allocations now playing the role that "zero-allocations'' played in \Cref{thm:downward}. 
Recall that a swap allocation for agent~$i$, denoted $\swapfeasible(i)$, is a feasible allocation vector with the property that if we replace the $i^{\text{th}}$ component of this vector with the $i^{\text{th}}$ component from any other feasible allocation, the resulting allocation is still feasible.

For any DSIC $(\allocs, \pays)$ on~$\typespaces$, 
We give a DSIC extension $(\extallocs, \extpays)$ on $\convhull(\typespaces)$. 

 Whenever the reported type profile $\types$ contains only types from the original typespace $\typespaces$ (i.e. $\types \in \typespaces$),
$\extallocs(\types)=\allocs(\types)$ and
$\extpays(\types)=\pays(\types)$,
as required by definition of an extension.

If only one agent~$i$ has a type $\typei \in \convhull(\typespacei)\setminus \typespacei$ and  their opponents have types $\typesmi \in \typespacesmi$, then agent~$i$ is given the same allocation as the type $\typei^\text{max} = \argmax_{\typei'\in \typespacei} \langle \typei, \alloci(\typei',\typesmi) \rangle - \payi(\typei', \typesmi)\rangle$. 
More formally, for $\typei \in \convhull(\typespacei)\setminus \typespacei$ and $\typesmi \in \typespacesmi$, the extended allocation rule for $\typei$ is
$
  \extalloci(\typei,\typesmi) =
  \alloci(\typei^\text{max},\typesmi)  
$
and the extended payment rule is
$
  \extpayi(\typei,\typesmi) =
  \payi(\typei^\text{max},\typesmi) 
$.

 The allocations for the remaining agents are described by agent~$i$'s swap feasible allocation, and
 are charged a payment of $0$.
 Formally, $\forall j\neq i$, $\extalloci[j](\typei[j], \typesmi[j]) = \swapfeasible_{j}(i)$ and $\extpayi[j](\typei[j], \typesmi[j]) = 0$.

By definition of single swap feasibility, regardless of the value of $\extalloci(\typei,\typesmi)$, $(\extalloci(\typei,\typesmi),\swapfeasible_{-i}(i))$ remains feasible.  

Finally, 
if a (reported) type profile $\types$ has more than one agent with a type not in the original typespace, among agents whose types are not within their original typespaces, let $j$ and~$k$ be the smallest and second smallest indices among them.
The extended allocation for such a type profile $\types$ is $\extallocs(\types)=(\swapfeasible_{j}(k),\swapfeasible_{-j}(j))$ and $\forall i,$ $\extpayi(\types)=0$. 
In other words, every agent's allocation is determined by the swap allocation of agent~$j$, except agent~$j$ himself, whose allocation is given by the swap allocation of agent~$k$.
Since this allocation only changes the $j^{\text{th}}$ component of $\swapfeasible(j)$, it is feasible by definition of SSF.

It remains to show that $(\extallocs, \extpays)$ is DSIC. 
First we deal with the case where an agent~$i$ has at least one opponent whose type is not in the original type space. 
Let $j$ be the smallest index of such an opponent. 
Then agent~$i$ is given $\swapfeasible_{i}(j)$ regardless of their reported type. 
Therefore we trivially satisfy DSIC when an opponent type profile contains a ``new'' type. 

When $\typesmi \in \typespacesmi$ we must show that any agent $i$ has no incentive to deviate, both when they have an ``old'' type and when they have an ``new'' type. 
For an ``old'' type $\typei \in \typespacei$, deviating to some type $\typei' \in \convhull(\typespacei)\setminus \typespacei$ gives the same utility as deviating to $\typei^\text{max} = \argmax_{\typei''\in \typespacei} \langle \typei',
\alloci(\typei'',\typesmi) \rangle - \payi(\typei'', \typesmi)$. 
Since $(\allocs, \pays)$ is DSIC, such a deviation cannot provide higher utility than $\langle \typei, \extalloci(\typei,\typesmi) \rangle - \extpayi(\typei, \typesmi)$. 
Finally, if $\typei \in \convhull(\typespacei)\setminus \typespacei$, then by definition of~$\extallocs$, by reporting truthfully the agent is given the allocation and payment that maximizes their utility over all possible deviations in $\typespacei$; by the logic above, this maximizes utility over all possible deviations in $\convhull(\typespacei)$ as well. Therefore, the agent has no incentive to deviate.
\end{proof}
\begin{figure*}
    \centering
    \includegraphics[scale=0.4]{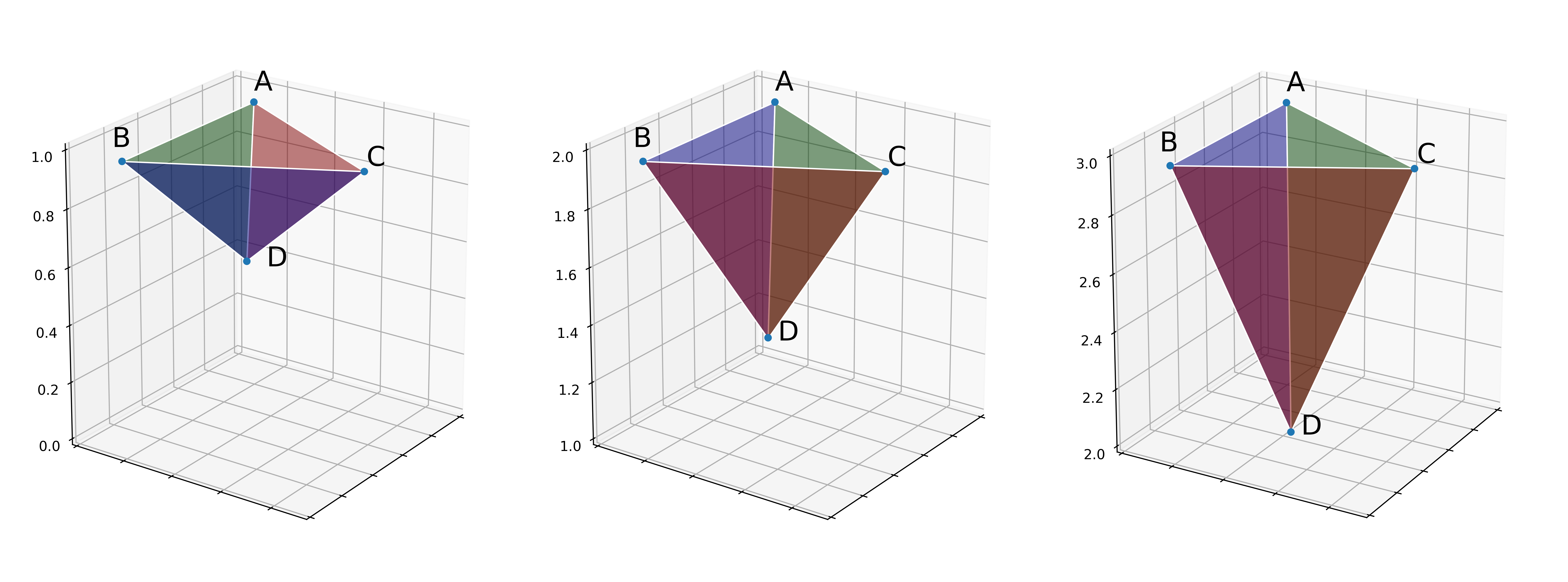}
    \caption{The maximum value for allocations in $\detalloci[1]$ for all types in $\convhull(\dettypespacei[1])$ . Each plot (visually an upside-down hollow pyramid) represents a reported type for agent 2. From left to right these represent agent $2$'s type being $A$, $B$ or $C$ ($\typei[2]=D$ is omitted). Within each plot, or set of planes, the color represents the type in $\dettypespacei$ agent $1$ should report to receive the highest value. The green plane represents $\typei[1]=A$, the blue plane represents $\typei[1]=B$ and the red plane represents $\typei[1]=C$. For example, the green plane in the first set of planes represents the value from allocation $\detalloci[1](A,A)$. From this image we can gain the intuition behind the ``locked'' payments by re-imagining these planes as describing utility when payments are zero.  We can think about changes in payments as vertical shifts in the planes. For example, changing $\payi[1]^{\text{det}}(A,A)$ would shift the plane corresponding to $\detalloci[1](A,A)$. This would either cause it to fall below the red plane at the point $A$ giving $A$ incentive to deviate to type $B$ or cause it to rise above the blue plane at point $B$ giving $B$ an incentive to deviate to type $A$.}
    \label{fig:values}
\end{figure*}
\begin{figure*}
    \centering
    \includegraphics{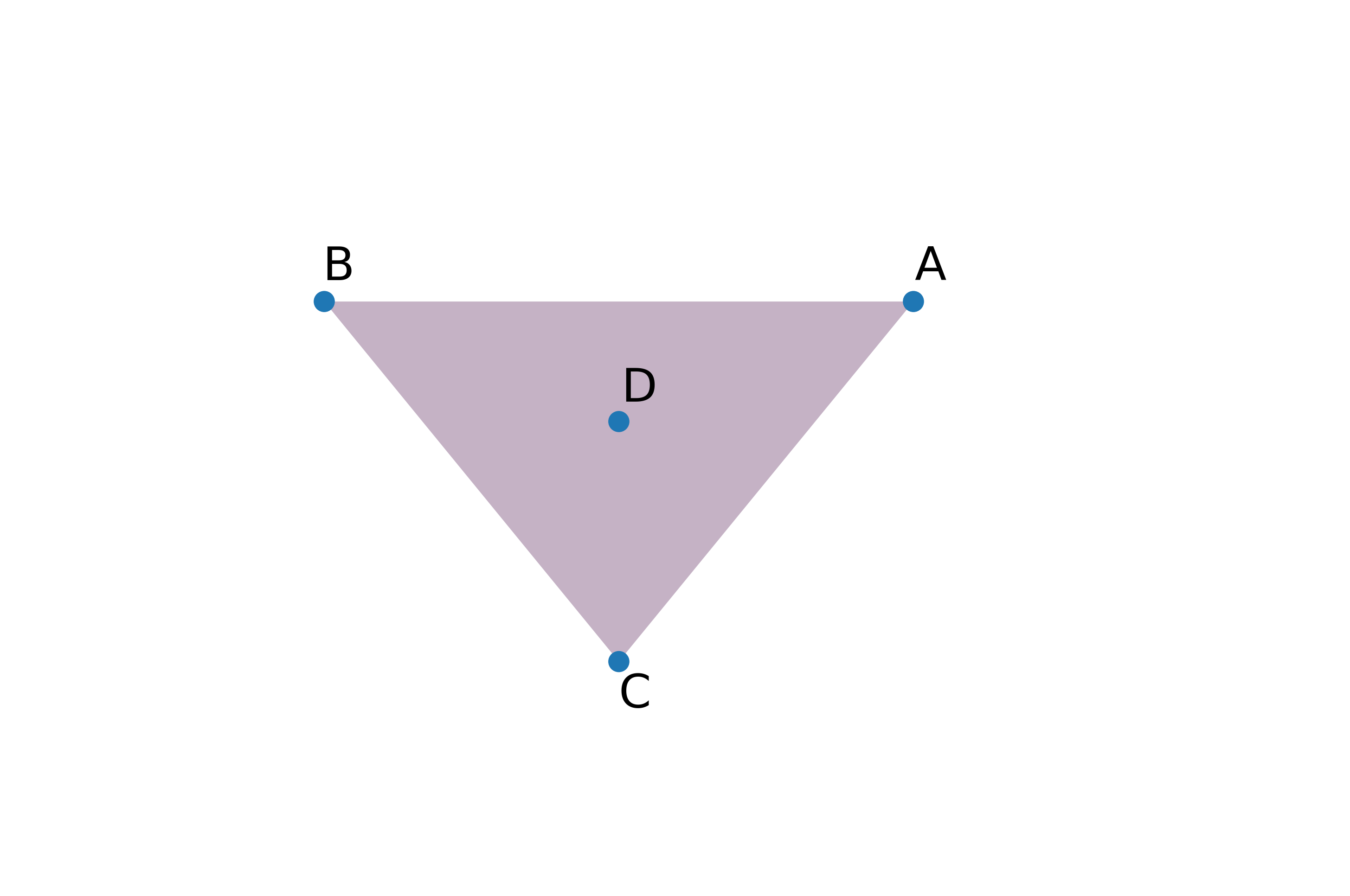}
    \caption{The typespace $\convhull(\dettypespace)$ projected to a $2$ dimensional plane. The types in $\dettypespace$ are labeled.}
    \label{fig:typespace}
\end{figure*}

\end{document}